\documentclass[12pt]{article}

\usepackage{amssymb,amsmath,amsthm,amsfonts}
\usepackage[colorlinks=true,linktocpage=true,allcolors=blue]{hyperref}
\usepackage{dsfont}
\usepackage{graphicx}
\usepackage{cite}
\usepackage{calc,ifthen,setspace,amscd}
\usepackage{bbm}
\usepackage{cleveref}
\usepackage{xcolor}
\usepackage{orcidlink}

\crefformat{equation}{(#2#1#3)}
\crefrangeformat{equation}{(#3#1#4)-(#5#2#6)}
\crefmultiformat{equation}{(#2#1#3)}{ and~(#2#1#3)}{, (#2#1#3)}{ and~(#2#1#3)}

\newtheorem{theorem}{Theorem}
\theoremstyle{plain}
\crefname{theorem}{Theorem}{Theorems}

\newtheorem{corollary}{Corollary}[theorem]
\crefname{corollary}{Corollary}{Corollaries}

\newtheorem{lemma}{Lemma}
\crefname{lemma}{Lemma}{Lemmas}

\newtheorem*{definition}{Definition}
\crefname{definition}{Definition}{Definitions}

\newtheorem{proposition}{Proposition}
\crefname{proposition}{Proposition}{Propositions}

\newtheorem{conjecture}{Conjecture}
\crefname{conjecture}{Conjecture}{Conjectures}

\crefname{appendix}{Appendix}{Appendices}

\crefname{section}{Section}{Sections}

\title{ \bf{All toric K\"{a}hler surfaces with twistor 2-forms}}

\date{}
\author{Sergei G. Ovchinnikov\, \orcidlink{0000-0003-2822-755X}\footnote{sovchinnikov@itmp.msu.ru}
\\ \small \sl Institute of Theoretical and Mathematical Physics, \\ \small \sl Lomonosov Moscow State University, Moscow 119991, Russia}

\begin{document}
    \maketitle

\begin{abstract}
We complete the classification of all smooth 4-dimensional K\"{a}hler geometries admitting a twistor (conformal Killing-Yano) 2-form invariant under a 2-torus action. We establish that there are six geometrically distinct families, and we provide them in a simple form amenable to calculations and compute their curvature. We also find that for toric geometries the square norm of the twistor 2-form is quadratic in moment maps, and we are led to conjecture that this holds when less symmetry is present.
\end{abstract}

\section{Introduction}

K\"{a}hler geometries are of great interest to mathematical physics. For instance, they naturally arise in string theory and supergravities as a structural part of gravitational backgrounds that are supersymmetric in the sense of admitting a real Killing spinor \cite{gauntlett2003,gauntlett2004,cassani2016,Gauntlett:2002fz,gauntlett2008,Gauntlett:2004zh}. In such settings, their curvature is typically further additionally constrained by a class of geometric PDEs which are usually too complicated to be solved both explicitly or numerically. Naturally the subclasses of K\"{a}hler geometries that are amenable to calculations are in great demand, and the primary candidates are those that carry additional twistor structures because of their intrinsic relation to integrability and separability of differential equations. In fact, in many cases, the only progress was achieved precisely with such subclasses. For $d=5$ minimal gauged supergravity, it has been observed \cite{cassani2016,blazquez-salcedo2018} that 4d K\"{a}hler geometries of all known (time-like) supersymmetric solutions possess hamiltonian 2-forms, which led to a progress in showing uniqueness of certain classes of gravitational backgrounds \cite{lucietti_ovchinnnikov_su2_2021,lucietti2022,Lucietti:2023mvj}. More generally, one considers backgrounds of the form $AdS_{n} \times Y_{m}$ where $Y_{m}$ is K\"{a}hler or Sasaki depending on dimension $m$, whose curvature obeys a PDE expressing the integrability of a Killing spinor equation; these correspond to brane-like configurations in different supergravities, and the search for such solutions is a rapidly developing field \cite{Gauntlett:2004hh,Martelli:2004wu,Gauntlett:2019pqg,ferrero2022a,Faedo:2024upq}. Another application are various extremisation principles put forward in mathematical \cite{Futaki:2006cc,Li:2016qwl,Futaki:2017vup} and physical literature \cite{Martelli:2005tp,couzens2019,Gauntlett:2019roi,boido2022,martelli2023,BenettiGenolini:2023kxp}. In the latter setting, these principles generally\footnote{With the exception of volume extremisation for Sasaki-Einstein manifolds \cite{Martelli:2005tp}.} provide necessary but not sufficient conditions for the existence of supersymmetric solutions, and geometries with twistor structures naturally provide testable grounds for the validity and completeness of these principles.

The first major step in classification of K\"{a}hler geometries with twistor structures was the introduction of hamiltonian 2-forms \cite{apostolov2003}, which are Hermitian (J-invariant) closed forms constructed on K\"{a}hler manifolds by adding a multiple of the K\"{a}hler form to a twistor 2-form. These forms are closely related to Killing vector fields and their properties have been extensively studied in the literature\footnote{Geometries with hamiltonian 2-forms were also widely used for construction of new K\"{a}hler and Sasaki geometries, see, e.g., \cite{Martelli:2005wy,legendre2011,Bykov:2017mgc,apostolov2024}.} in various dimensions \cite{apostolov2004,apostolov2004a,apostolov2008,madani2019}. Importantly, in \cite{apostolov2004a} it was shown that all 4d geometries admitting hamiltonian 2-forms break into three distinct families — the bik\"{a}hler, the Calabi type and the orthotoric, each defined up to two arbitrary functions of a single variable. As it turns out, not all twistor 2-forms allow for hamiltonian 2-forms, and the necessary conditions in terms of curvature were found in \cite{apostolov2004}. For toric K\"{a}hler surfaces, we show that the necessary conditions are also sufficient:

\begin{lemma}
    Let $(M,g,J)$ be a 4d toric K\"{a}hler geometry with a SD twistor 2-form $\phi^{tw}$ that is invariant under the torus symmetry and $\iota_{m_1} \iota_{m_2} \phi^{tw} = 0$, that is $\phi^{tw}$ belongs to the negative eigenspace of the inversion of angles $\mathrm{i}_\varphi$.\\ Then $(M,g,J)$ admits a hamiltonian 2-form $\phi = \phi^{tw} + \sigma J$ for some $\sigma$ if and only if the SD part of the Ricci form is a multiple of the twistor 2-form
    \begin{equation*}
        \phi^{tw} \propto \mathcal{R}^+ =\mathcal{R} - \frac{R}{4} J\,.
    \end{equation*}
\end{lemma}

K\"{a}hler surfaces admitting additional twistor 2-forms can be studied in the framework of ambik\"{a}hler structures \cite{Apostolov:2013oza,Apostolov:2013lea}. Recall that a twistor 2-form is equivalent to a locally conformally K\"{a}hler structure \cite{pontecorvo1992}. Ambik\"{a}hler structures are then defined as pairs of 4d K\"{a}hler metrics in the same conformal class whose K\"{a}hler structures introduce opposite orientations. If both K\"{a}hler geometries in a pair are invariant under a common torus action, the geometry is known as ambitoric, and in the same paper, they were shown to belong to one of the five types, each parameterised by two functions of a single variable and a number of real constants. These types are pairs of toric K\"{a}hler products and toric Calabi geometries as well as three regular ambitoric structures: elliptic, parabolic and hyperbolic. A priori, however, it is not clear whether individual K\"{a}hler metrics in the pairs are geometrically distinct in the sense of being diffeomorphic, perhaps up to a redefinition of free functions and constants. By providing an alternative construction for ambitoric geometries, we show that in all pairs except the parabolic type, the geometries can indeed be obtained by such an isomorphism, while for the parabolic type, the geometries are distinct for any parameters; we then suggest a convenient chart for each geometry. We further show that that if the conformally K\"{a}hler structure associated with a twistor 2-form does not have a common torus action, then the geometry cannot be smooth. Our main result is the following theorem. 

\setcounter{theorem}{0}
\begin{theorem}\label{thm:main}
    Let $(M,g,J)$ be a smooth toric K\"{a}hler surface, and $\phi^{tw}$ a self-dual twistor 2-form on it, which is invariant under the torus symmetry. We further assume that the axis is non-empty, i.e. there is a point where a linear combination of axial Killing fields degenerate.\\
    Then locally the geometry is of the form 
    \begin{equation*}
        \mathrm{d} s^2 = \dfrac{e^\mu}{F} \mathrm{d} \xi^2 \,+\, \dfrac{e^\mu}{G} \mathrm{d} \eta^2 \,+\, \dfrac{F}{e^\mu} \left(x_\xi \mathrm{d} \Psi\,+\, y_\xi \mathrm{d} \Phi \right)^2 \,+\,\dfrac{G}{e^\mu} \left(x_\eta \mathrm{d} \Psi\,+\, y_\eta \mathrm{d} \Phi\right)^2
    \end{equation*}
    where $F=F(\xi)$, $G = G(\eta)$ are two arbitrary functions, $x,y$ are moment maps wrt axial Killing fields $m_\Psi=\frac{\partial}{\partial \Psi},$ $m_\Phi = \frac{\partial}{\partial \Phi}$, and $e^\mu$ is the norm of the twistor 2-form \cref{eq:constr_tw-form}. The geometric data $(\mu, x, y)$ belongs to one of six independent families: product-toric (bik\"{a}hler) \cref{eq:twist-spec_Product-toric}, Calabi-toric \cref{eq:twist-spec_Calabi-toric}, orthotoric \cref{eq:twist-spec_Orthotoric}, elliptic \cref{eq:twist-gen-sol-elliptic}, conformally orthotoric (parabolic) \cref{eq:twist-spec_conf-ortho} and hyperbolic \cref{eq:twist-gen-sol-hyperbolic}.
    For the first three families, their twistor 2-forms give rise to hamiltonian 2-forms, and these comprise all toric hamiltonian 2-form geometries. 
\end{theorem}

Finally, the explicit form of the geometries allows us to make an interesting observation — we notice that for all families, the norm squared of the twistor 2-form is at most quadratic in moment maps. We therefore put forward a conjecture that this should hold in general, which would allow for the generalisation of our result to geometries outside of the toric class.
\begin{conjecture}\label{conj:conj}
    Let $(M,g,J)$ be a 4d K\"{a}hler geometry with a twistor 2-form $\phi^{tw}$ invariant under a hamiltonian action of a Lie group $G$, and let $e^\mu$ be the norm of $\phi^{tw}$.\\
    Then the norm squared $e^{2 \mu}$ of the twistor 2-form is at most a quadratic polynomial in moment maps corresponding to the maximal torus algebra of $G$.
\end{conjecture}

The paper is organised as follows: in \cref{sec:twistor_kahl} we present necessary background material which includes some basics on toric K\"{a}hler surfaces and twistor 2-forms and show that for smooth K\"{a}hler geometries the twistor 2-form must generate the same torus action. In \cref{sec:toric_kahl} we first provide an alternative construction for ambitoric geometries and then prove our classification result with the help of our observation \cref{prop:square_norm is quadratic}. \cref{sec:hamiltonian_2_forms} introduces an important class of geometries that admit a twistor 2-form. A novel result are the necessary and sufficient conditions on curvature from the existence of hamiltonian 2-forms for toric K\"{a}hler geometries. Finally, in \cref{app:ricci_curvature_for_geometries_with_twistor_2_form} we provide the Ricci and scalar curvature for new geometries and in \cref{app:solution_of_the_ODEs} we give detailed solutions of the ODEs from \cref{sec:toric_kahl}.

\section{Twistor 2-forms and K\"{a}hler geometries} 
\label{sec:twistor_kahl}

\subsection{K\"{a}hler geometry}
\label{sub:kahl_geometry}
A \textit{K\"{a}hler surface} $(M,g,J)$ is an orientable Riemannian four-dimensional manifold together with an \textit{anti-self-dual} complex structure $J$, such that $\nabla J =0$ where $\nabla$ is the Levi-Civita connection wrt $g$. With a certain abuse of notation, we are denoting the almost complex structures and their respective 2-forms by the same letter. 

A \textit{toric K\"{a}hler surface} is a K\"{a}hler surface together with an effective hamiltonian action of the real torus $T^2 \cong U(1)^2$; we are interested in smooth K\"{a}hler geometries, and for these we include a \textit{non-empty axis} as part of the definition. This means that the geometry admits two Killing vector fields $m_i$, $i=1,2$, normalised to have $2 \pi$ periodic orbits, and the K\"{a}hler form is invariant under their action $\mathcal{L}_{m_i} J =0$\,. The hamiltonian action allows us to introduce moment maps $x_i$
\begin{equation}
    \iota_{m_i} J = - \mathrm{d} x_{i}\,.
\end{equation}
The moment maps define a canonical coordinate system $(x_i, \varphi^i)$ where $m_i = \partial_{\varphi^i}$. In this chart \cite{abreu2001},
\begin{align}\nonumber
    &g = G^{ij} \mathrm{d}x_i \mathrm{d}x_j \,+\, G_{ij} \mathrm{d}\varphi^i\,\mathrm{d}\varphi^{j}\,,\\\label{eq:kahl_geometry-h-Gij-J}
    &G^{ij} = \partial^{i} \partial^{j} \mathrm{g}\,,\\\nonumber
    &J = \mathrm{d}\left(x_i \mathrm{d} \varphi^{i}\right)\,,
\end{align}
where $\mathrm{g} = \mathrm{g}(x)$ is the symplectic potential, $G_{ij}$ is the inverse matrix of the Hessian $G^{ij}$ and the derivatives stand for $\partial^i := \partial/\partial x_{i}$. The axes are defined by $\{p \in M |\mathrm{det}\, g(m_i, m_j)|_p = 0\}$\,, and we assume it to be non-empty. Note from \cref{eq:kahl_geometry-h-Gij-J}, that this Gram determinant is
\begin{equation}
    \mathrm{det}\, g(m_i, m_j) = \mathrm{det}\, G_{ij} =: G_{gr}\,.
\end{equation}

One can introduce a basis of anti-self-dual (ASD) and self-dual (SD) forms by picking $(J^1 =J, J^2, J^3)$ for ASD forms 
\begin{align}
    &J^2 = G_{gr}^{-1/2}\, \mathrm{d} x_1 \wedge \mathrm{d} x_2 - G_{gr}^{1/2} \,\mathrm{d} \varphi^1 \wedge \mathrm{d} \varphi^2\,,\\
    &J^3 = -2 G_{gr}^{1/2} \, G^{i [1} \mathrm{d} x_i \wedge \mathrm{d} \varphi^{2]}\,,
\end{align}
and for SD forms
\begin{align}\label{eq:kahl_I1}
    &I^1 = \mathrm{d} x_1 \wedge \mathrm{d} \varphi^1 + \frac{2 G^{12}}{G^{11}} \mathrm{d} x_2 \wedge \mathrm{d} \varphi^1 - \mathrm{d} x_2 \wedge \mathrm{d} \varphi^2 \,,\\\label{eq:kahl_I2}
    &I^2 = G_{gr}^{-1/2}\, \mathrm{d} x_1 \wedge \mathrm{d} x_2 + G_{gr}^{1/2}\,\mathrm{d} \varphi^1 \wedge \mathrm{d} \varphi^2\,,\\\label{eq:kahl_I3}
    &I^3 = G_{gr}^{1/2} (G^{12} \mathrm{d} x_1 \wedge \mathrm{d} \varphi^1 - G^{11} \mathrm{d} x_1 \wedge \mathrm{d} \varphi^2 \\ \nonumber
    &\hphantom{I^3 =} - \frac{G_{gr}^{-1} - \left( G^{12} \right)^2}{G^{11}} \mathrm{d} x_2 \wedge \mathrm{d} \varphi^1 - G^{12} \mathrm{d} x_2 \wedge \mathrm{d} \varphi^2)\,.
\end{align}
Toric geometries naturally admit an involution given by the inversion of angles
\begin{equation}\label{eq:kahl_geometry_angular-involution}
    \mathrm{i}_\varphi: \quad \varphi^i \mapsto - \varphi^{i}\,,
\end{equation}
and the basis above respects it in the sense that the positive eigenspace of $\mathrm{i}_\varphi$ is spanned by $J^2, I^2$, and the rest, including the K\"{a}hler form, span the negative eigenspace. Observe, that the negative eigenspace is the space of all 2-forms obeying $\iota_{m_1} \iota_{m_2}\Omega = 0$. Therefore, one can further show that any smooth closed invariant under the toric symmetry 2-form $\Omega$ on $M$ must belong to it. Indeed, from closure and toric invariance $\iota_{m_1} \iota_{m_2}\Omega$ is a constant, which, further, must vanish on the axis. Hence $\iota_{m_1} \iota_{m_2}\Omega = 0$ for any such smooth 2-form.

Finally, the behaviour at the axis is given by the following lemma \cite{lucietti2022}.
\begin{lemma} \label{lem:tor_axis}
Consider a neighbourhood of a component of the axis defined by the vanishing of $v:= v^i m_i$, where $(v^i )\in \mathbb{Z}^2$ are coprime integers, and let
\begin{equation}
\ell_v(x):= v^i x_i +c_v  \; , \label{eq:tor_axisline}
\end{equation}
where $c_v$ is a constant. Then:
\begin{enumerate}
\item  The axis component corresponds to a straight line $\ell_v(x)=0$ in symplectic coordinates and away from the axis
\begin{equation}
\ell_v(x)>  0
\end{equation}
\item The symplectic potential can be written as
\begin{equation}
\mathrm{g}= \frac{1}{2}\ell_v(x) \log \ell_v(x) + \tilde{\mathrm{g}}  \; ,
\end{equation}
where $\tilde{\mathrm{g}}$ is smooth at $\ell_v(x)=0$.
\end{enumerate}
\end{lemma}

\subsection{Twistor forms and K\"{a}hler metrics}
\label{sub:Tw&K}

Twistor 2-forms, also known as conformal Killing-Yano 2-forms, are a generalisation of conformal Killing vectors (seen as twistor 1-forms). Namely, on a Riemannian manifold $(M,g)$ the orthogonal decomposition of a covariant derivative of a general $p$-form is
\begin{equation}
    \nabla \phi \in \Lambda^1 M \otimes \Lambda^{p} M = \Lambda^{p+1}M \oplus \Lambda^{p-1}M \oplus \mathcal{T}^{p,1}M
\end{equation}
where the first two terms in the sum are the exterior derivative and the co-differential bundles, and the last term $\mathcal{T}^{p,1}M$ is the \textit{Cartan product}. For 1-forms, $\mathcal{T}^{1,1}M \cong S_0^2 M$ is isomorphic to the bundle of trace-free symmetric 2-tensors, and the conformal Killing vector equation is equivalent to stating that the projection on it vanishes. Similarly for general $p$, the \textit{twistor p-form} is the form $\phi$ such that
\begin{equation}
    \nabla \phi \in \Lambda^{p+1}M \oplus \Lambda^{p-1}M\,.
\end{equation}
It is convenient for us to reformulate this definition more explicitly.
\begin{definition}\label{def:twistor-2form}
    Let $(M,g)$ be a $4d$ Riemannian manifold with a Levi-Civita connection $\nabla$. The twistor 2-form is a non-trivial solution to the following equation\footnote{In index notation this reads as \begin{equation*}
    \nabla_m  \phi^{tw}_{np} = \partial_{[m}  \phi^{tw}_{np]} + \frac{2}{3} g_{m[n} (\nabla  \phi^{tw})_{p]}\,.
\end{equation*}} $\forall X \in TM$
    \begin{equation}\label{eq:Tw&K_twistor-definition}
    \nabla_X \phi^{tw} = \frac{1}{3} \iota_X \mathrm{d} \phi^{tw} + \frac{1}{3} X \wedge \delta  \phi^{tw}\,.
\end{equation}
\end{definition}

Twistor 2-forms form a linear space which further decomposes under any involution of $\Lambda^2 M$ that commutes with the connection into a direct sum of its eigenspaces. One such involution is the Hodge duality itself, which allows to consider SD and ASD twistor forms separately. In the following, we will write twistor forms as $\phi^{tw} = e^{\mu} I$ where $I$ is a unit SD (ASD) form, i.e. $I \wedge I = \pm 2 \mathrm{vol}_{g}$, and $\mu$ is a continuously differentiable function.

Twistor 2-forms are linked to existence of a locally conformally K\"{a}hler structure \cite{pontecorvo1992} (see also \cite[Lemma 2]{apostolov2003}).
\begin{lemma}{\cite{pontecorvo1992}}\label{lem:Pontecorvo}
    The SD (ASD) form $\phi^{tw} = e^{\mu} I$ is a twistor 2-form on $(M,g)$ if and only if the almost-Hermitian pair\footnote{That is, a Hermitian inner product $e^{-2 \mu} g$ and an almost complex structure $e^{-2 \mu} I$.} $(e^{-2 \mu} g\,,e^{-2 \mu} I)$ is K\"{a}hler.
\end{lemma}
This gives an alternative form of the twistor 2-form equation.
\begin{corollary}{\cite{apostolov2003}}\label{cor:alt-form-of-tw-eqn}
    The SD (ASD) form $\phi^{tw} = e^{\mu} I$ on $(M,g)$ is twistor if and only if the following equation holds
    \begin{gather}\label{eq:corollary-nabla_I}
        \nabla I = \gamma \otimes I^2 \,-\, \beta \otimes I^3\,,\\ \nonumber
        \beta = I^2 \mathrm{d}\mu\,, \qquad \gamma = I^3 \mathrm{d}\mu
    \end{gather}
    where $(I^1 = I , I^2, I^3)$ is an orthonormal basis of SD (ASD) forms which satisfies standard quaternionic relations: $I^1 I^2 = I^3$, etc.    
\end{corollary}

Besides Hodge duality, K\"{a}hler geometries have another natural involution. One can define a $J-$conjugation $\mathrm{i}_J \phi (X,Y) := \phi(JX, JY)$ which further decomposes ASD subspace\footnote{SD subspace does not decompose because all SD forms are J-invariant by construction.} into a direct sum of a one-dimensional space generated by the K\"{a}hler form and its orthogonal complement, anti-J-invariant 2-forms. The K\"{a}hler form itself is a twistor form, and if a geometry admits a twistor form belonging to the latter subspace, then it is hyperk\"{a}hler.
\begin{proposition}
    Let $(M,g,J)$ a K\"{a}hler surface with ASD K\"{a}hler form $J$. Then the geometry admits an ASD twistor 2-form which is not a multiple of K\"{a}hler form if and only if the geometry is hyperk\"{a}hler.
\end{proposition}
\begin{proof}
    Write an anti-J-invariant twistor 2-form as $\phi^{tw} = e^\mu J^{tw}$, $J^{tw} \wedge J^{tw} = - 2 \mathrm{vol}_{g}$, and pick an orthonormal basis of ASD forms $(J^1 = J, J^2 = J^{tw}, J^3=J^1 J^2)$. By the property of K\"{a}hler geometry $\nabla J^{tw} = P \otimes J^3$ where $P$ is the Ricci potential, and by \cref{cor:alt-form-of-tw-eqn} it follows that $P=0$, i.e. the geometry is Ricci-flat.

    Conversely, since any K\"{a}hler form is a twistor form, a hyperk\"{a}hler geometry admits multiple twistor forms.
\end{proof}

Finally, for toric K\"{a}hler geometries, the inversion of angles \cref{eq:kahl_geometry_angular-involution} is the third involution that commutes with a connection (and other involutions), which for toric invariant twistor forms allows to consider the case $\phi^{tw} \propto I^2$ separately. One can show that geometries admitting such a SD twistor form cannot be smooth.
\begin{proposition}\label{prop:no-I2}
    Let $(M,g,J)$ be a smooth toric K\"{a}hler surface. Then it does not admit a toric invariant SD twistor 2-form that belongs to the positive eigenspace of $\mathrm{i}_\varphi$\,.
\end{proposition}
\begin{proof}
    Let $\phi^{tw} = e^\mu I^2$ be a toric invariant twistor form. Firstly, by \cref{lem:Pontecorvo} $e^{-2 \mu} I^2$ is a closed form, hence, using \cref{eq:kahl_I2},
    \begin{equation}
        \mu = \frac{1}{4} \mathrm{log}\, G_{gr} + \textrm{const}.
    \end{equation} 
    Furthermore, adapting the \cref{cor:alt-form-of-tw-eqn}, it follows that
    \begin{gather}\label{eq:kahl-toric-twistor-I2}
        \nabla I^2 = \alpha \otimes I^3 \,-\, \gamma \otimes I^1\,,\\ \nonumber
        \alpha = I^1 \mathrm{d}\mu\,, \qquad \gamma = I^3 \mathrm{d}\mu\,.
    \end{gather}
    Notice that this must hold for any choice of local basis forms $I^1, I^3$ which are unit and satisfy the quaternionic relations, and it is convenient for us to check twistor conditions in the basis \cref{eq:kahl_I1,eq:kahl_I2,eq:kahl_I3}. 

    Secondly, observe that 1-forms $\alpha, \gamma$ as well as $\beta$ define a connection on the bundle of SD forms
    \begin{gather}
        \nabla I^1 = \gamma \otimes I^2 \,-\, \beta \otimes I^3\,,\\ \nonumber
        \nabla I^3 = \beta \otimes I^1 \,-\, \alpha \otimes I^2\,\hphantom{,}
    \end{gather}
    where we have used their orthonormality, and for any toric K\"{a}hler geometry in the symplectic chart, they read as
    \begin{align}\label{eq:kahl-alpha-tilde}
        \tilde{\alpha} &:= -I^1 \alpha = \partial^1 \mathrm{log}\, \left(G_{gr}^{1/2} G^{11}\right) \mathrm{d} x_1 + \tilde{\alpha}^2 \mathrm{d} x_2\,,\\\label{eq:kahl-gamma-tilde}
        \tilde{\gamma} &:= -I^3 \gamma = - \partial^1 \mathrm{log}\, \left(G_{gr} \ G^{11}\right) \mathrm{d} x_1 + \tilde{\gamma}^2 \mathrm{d} x_2
    \end{align}
    where $\tilde{\alpha}^2, \tilde{\gamma}^2$ are some functions irrelevant to discussion. 
    
    To obtain a contradiction, consider a smooth toric K\"{a}hler geometry in the neighbourhood of the axis given by vanishing of $v:= v^i m_i$. By \cref{lem:tor_axis} the axes correspond to straight lines in symplectic coordinates, and by a suitable $GL(2, \mathbb{Z})$ transformation we can always arrange $v := m_1 = \partial_{\varphi^1}$ and define new symplectic coordinates so that the axis corresponds to the line $\ell_v(x) = x_1 =0$. As it was further shown in this lemma, to the leading order the geometry around the axis is
    \begin{equation}
        G^{11} = \frac{1}{2 x_1} + O(1)\,, \qquad G^{12},\,G^{22} = O(1)\,, \qquad G_{gr} = 2 x_1 + O(x_1^2)\,.
    \end{equation}

    Substituting into the twistor 2-form condition $\tilde{\alpha} = \tilde{\gamma} = \mathrm{d} \mu$ and using \cref{eq:kahl-alpha-tilde,eq:kahl-gamma-tilde}, we arrive at the contradiction
    \begin{equation}
        \tilde{\alpha} = \frac{1}{2 x_1}\mathrm{d} x_1 + O(1)\,, \qquad \tilde{\gamma} = O(1)\,.
    \end{equation}
    
\end{proof}

The case of twistor forms in the positive eigenspace of $\mathrm{i}_\varphi$ is, therefore, excluded from our classification, and we consider the case of a negative eigenspace in the next section.

\section{Toric K\"{a}hler geometries with twistor 2-forms} 
\label{sec:toric_kahl}

In this section we solve the twistor 2-form equation on toric K\"{a}hler geometries by providing an alternative construction of ambitoric geometries \cite{Apostolov:2013oza} an appropriate chart. 
We will show that the geometries fall into separate families each containing two arbitrary functions $F, G$ of a single variable. We are interested in classifying them up to \textit{family automorphisms}, that is up to diffeomorphisms and a change of metric functions $F,G \rightarrow \tilde{F}, \tilde{G}$.

\subsection{Construction of a diagonalised chart}
\label{sec:construction_of_diag_chart}

1. Let $(M, g, J)$ be a toric K\"{a}hler surface, and $x,y$ are moment maps wrt commuting Killing vector fields $m_\Psi=\frac{\partial}{\partial \Psi},$ $m_\Phi = \frac{\partial}{\partial \Phi}$ whose periods are left arbitrary. The K\"{a}hler form is then 
\begin{equation}\label{eq:constr_J}
   J = \mathrm{d} \left(x \mathrm{d} \Psi + y \mathrm{d} \Phi\right)\,. 
\end{equation}
Let the geometry admit a twistor 2-form $\phi^{tw} = e^\mu I$ invariant under the toric symmetry, and which belongs to the negative eigenspace of $\mathrm{i}_\varphi$, that is $\iota_{m_\Psi}\iota_{m_\Phi} \phi^{tw} = 0$, where $I$ is the unit SD form $I \wedge I = 2 \mathrm{vol}_{g}$, and $\mu$ is an unknown function. We can further write it as
\begin{equation}\label{eq:constr_I+J}
    I = - J + \omega_\Psi \wedge \mathrm{d} \Psi +  \omega_\Phi \wedge \mathrm{d} \Phi
\end{equation}
for some non-zero 1-forms $\omega_\Psi, \omega_\Phi$ normal to the orbits of $m_\Psi, m_\Phi$. Notice that $I+J$ is necessarily degenerate\footnote{We use the orthogonality of SD and ASD 2-forms: $\star(I \wedge J) = \star(\star I \wedge J) = \star(I \wedge \star J) = -\star(I \wedge J)=0$.} $$(I+J) \wedge (I+J) = I \wedge I + J \wedge J = 2 \mathrm{vol}_{g} - 2 \mathrm{vol}_{g}=0\,,$$ hence
\begin{equation}
    \omega_\Psi \wedge \omega_\Phi = 0\,,
\end{equation}
and we can write them as $\omega_\Psi = k_\Psi \mathrm{d} \xi$\,, $\omega_\Phi = k_\Phi \mathrm{d} \xi$ for some functions $k_\Psi,k_\Phi,\xi$ defined up to the reparameterisations of the latter. 

2. Next, consider the 2 by 2 block of symplectic coordinates inside the metric. Recall, that we can\footnote{See, for example, \cite[Appendix C]{Harmark:2004rm}.} always diagonalise a two-dimensional Riemannian metric wrt any given function being one of the coordinates, provided that it has a non-vanishing gradient. Since $\mathrm{d} \xi \neq 0$, it is convenient to pick $\xi$ as one, and denote the second coordinate as $\eta$. 

3. Let us introduce the frame. Without loss of generality, we can take $E^\xi \propto \mathrm{d} \xi$ and choose the K\"{a}hler form as $J= E^\xi \wedge E^\Psi + E^\eta \wedge E^\Phi$. This leaves the $SO(2)$ freedom in the choice of $E^{\eta}, E^{\phi}$ which we can use to set $\iota_{\partial_\phi} E^{\eta} = 0$. Then the rest of the frame functions are fixed from checking the scalar products and the K\"{a}hler form \cref{eq:constr_J}. In total, we can parameterise the geometry by two arbitrary functions $F=F(\xi,\eta)$, $G=G(\xi,\eta)$, and the derivatives of the moment maps wrt new coordinates $x_\xi, x_\eta, y_\xi, y_\eta$\,:
\begin{equation}\label{eq:constr_frame}
    \begin{aligned}
        E^\xi  = \dfrac{e^{\mu/2}}{\sqrt{F}}\mathrm{d} \xi \,, \quad E^\Psi = \dfrac{\sqrt{F}}{e^{\mu/2}} \left(x_\xi \mathrm{d} \Psi\,+\, y_\xi \mathrm{d} \Phi\right)\,,\\
        E^\eta = \dfrac{e^{\mu/2}}{\sqrt{G}}\mathrm{d} \eta\,, \quad E^\Phi = \dfrac{\sqrt{G}}{e^{\mu/2}} \left(x_\eta \mathrm{d} \Psi\,+\, y_\eta \mathrm{d} \Phi\right)\,.
    \end{aligned}
\end{equation}
Finally, we pick the orientation $E^\xi \wedge E^\eta \wedge E^\Psi \wedge E^\Phi = \mathrm{vol}_{g}.$ The metric is then
\begin{equation}\label{eq:constr_metric}
    \mathrm{d} s^2 = \dfrac{e^\mu}{F} \mathrm{d} \xi^2 \,+\, \dfrac{e^\mu}{G} \mathrm{d} \eta^2 \,+\, \dfrac{F}{e^\mu} \left(x_\xi \mathrm{d} \Psi\,+\, y_\xi \mathrm{d} \Phi \right)^2 \,+\,\dfrac{G}{e^\mu} \left(x_\eta \mathrm{d} \Psi\,+\, y_\eta \mathrm{d} \Phi\right)^2\,.
\end{equation}
We can choose an orthonormal basis of SD forms
\begin{equation}\label{eq:constr_SD_basis}
    \begin{aligned}
        I^1 = E^\xi \wedge E^\Psi - E^\eta \wedge E^\Phi\,,\\
        I^2 = E^\xi \wedge E^\eta + E^\Psi \wedge E^\Phi\,,\\
        I^3 = E^\xi \wedge E^\Phi + E^\eta \wedge E^\Psi\,,
    \end{aligned}
\end{equation}
whose almost complex structures satisfy the standard quaternionic relations: $I^1\,I^2= I^3$, etc.

4. Next, we must check that $I$ is the unit self-dual 2-form. The equation \cref{eq:constr_I+J} is then
\begin{equation}
    I = - E^\xi \wedge E^\Psi - E^\eta \wedge E^\Phi + \dfrac{\sqrt{F}}{e^{\mu/2}} E^\xi \wedge \left(k_\Psi \mathrm{d} \Psi + k_\Phi \mathrm{d} \Phi\right)\,.
\end{equation}
Comparing to the basis \cref{eq:constr_SD_basis}, we immediately have $k_\Psi = 2 x_\xi$, $k_\Phi = 2 y_\xi$, and
\begin{equation}
    I = I^1 = x_\xi \mathrm{d} \xi \wedge \mathrm{d} \Psi + y_\xi \mathrm{d} \xi \wedge \mathrm{d} \Phi - x_\eta \mathrm{d} \eta \wedge \mathrm{d} \Psi - y_\eta \mathrm{d} \eta \wedge \mathrm{d} \Phi\,.
\end{equation}

5. We now turn to the twistor 2-form equation. First of all, by \cref{lem:Pontecorvo} it follows that $e^{-2 \mu} I$ is a closed form. This is equivalent to
\begin{equation}\label{eq:constr_xy_xh}
    \begin{aligned}
        x_{\xi \eta} = \mu_\xi x_\eta + \mu_\eta x_\xi\,,\\
        y_{\xi \eta} = \mu_\xi y_\eta + \mu_\eta y_\xi\,.
    \end{aligned}
\end{equation}

The remaining components of the twistor equation can then be written as
\begin{equation}\label{eq:constr_twist_xx_hh}
    \begin{aligned}
        x_{\xi \xi}   = x_\xi  \partial_\xi  \mathrm{log}\, y_\xi  + \dfrac{F_\eta G}{F^2 y_\xi} \sqrt{\mathrm{det} g}\,,\\
        x_{\eta \eta} = x_\eta \partial_\eta \mathrm{log}\, y_\eta - \dfrac{F G_\xi}{G^2 y_\eta} \sqrt{\mathrm{det} g}\,\hphantom{,}
    \end{aligned}
\end{equation}
where $\sqrt{\mathrm{det} g} = x_\xi y_\eta - x_\eta y_\xi$ is the square root of the determinant of the metric \cref{eq:constr_metric}.

6. We will now check that $(M,g,J)$ is a K\"{a}hler geometry. By construction $J$ is a closed unit ASD form. To check that it is an integrable complex structure, it is sufficient to show that the algebra of holomorphic fields is closed under standard commutator. Let $\chi_0, \chi_1$ be a basis of $(1,0)$-forms, i.e. $J \chi_{0,1} = i \chi_{0,1}$\,. We then require that $[\chi_0, \chi_1]$ is also a $(1,0)$-form. We choose
\begin{equation}
    \chi_0 = E_\xi - i J E_\xi = E_\xi + i E_\Psi\,, \qquad \chi_1 = E_\eta - i J E_\eta = E_\eta + i E_\Phi\,
\end{equation}
where $E_\xi, E_\eta, E_\Psi, E_\Phi$ are basis vector fields dual to the frame \cref{eq:constr_frame}. We find that
\begin{equation}
    [\chi_0,\chi_1] = \chi^\xi E_\xi + \chi^\eta E_\eta + \chi^\Psi E_\Psi + \chi^\Phi E_\Phi\,
\end{equation}
where 
\begin{equation}
    \chi^{\xi} = \dfrac{e^{-\mu/2} \sqrt{G}}{2F} (F \mu_\eta - F_\eta)\,,\qquad \chi^\eta = \dfrac{e^{-\mu/2} \sqrt{F}}{2 G} (G_\xi - G \mu_\xi)\,,
\end{equation}
and the other two components, after substituting the higher derivatives from the twistor equation \cref{eq:constr_xy_xh,eq:constr_twist_xx_hh},
\begin{equation}
    \chi^{\Psi} = i \,\dfrac{e^{-\mu/2} \sqrt{G}}{2F} (F \mu_\eta + 3 F_\eta)\,,\qquad \chi^\Phi = - i \dfrac{e^{-\mu/2} \sqrt{F}}{2 G} (3 G_\xi + G \mu_\xi)\,.
\end{equation}

The integrability condition $[\chi_0, \chi_1] = i J[\chi_0, \chi_1]$ in terms of components reads as
\begin{equation}
    \chi^\Psi = i \chi^\xi\,, \qquad \chi^\Phi = i \chi^\eta\,,
\end{equation}
from which immediately follows that the functions $F,G$ are functions of a single variable only
\begin{equation}\label{eq:constr_Kahl-sol}
    F_\eta = G_\xi = 0\,.
\end{equation}

7. We now turn back to the twistor equation, which we have so far solved algebraically for the higher derivatives. Using the solution of the K\"{a}hler condition \cref{eq:constr_Kahl-sol}, the system \cref{eq:constr_twist_xx_hh} simplifies to
\begin{equation}
    \frac{x_{\xi \xi}}{x_\xi} = \frac{y_{\xi \xi}}{y_\xi}\,,\qquad \frac{x_{\eta \eta}}{x_\eta} = \frac{y_{\eta \eta }}{y_\eta}\,,
\end{equation}
which is solved by
\begin{equation}\label{eq:constr_x=c1y+c2=c3y+c4}
    y = c_1(\xi) x + c_2(\xi) = c_3(\eta) x + c_4(\eta)\,,
\end{equation}
for $c_{1,2,3,4}$ arbitrary functions. The moment maps $x,y$ are then
\begin{equation} \label{eq:constr_xy-c1234_sol}
    x = \dfrac{c_4 - c_2}{c_1 - c_3}\,, \qquad y = \dfrac{c_1 c_4 - c_2 c_3}{c_1 - c_3}\,.
\end{equation}

This section can be summed up in a following lemma.
\begin{lemma}\label{lem:diag_chart}
    Let $(M,g,J)$ be a toric K\"{a}hler geometry with axial Killing fields $m_\Psi, m_\phi$, admitting a toric-invariant twistor 2-form $\phi^{tw}$ which belongs to the negative eigenspace of the inversion of angles $\mathrm{i}_\varphi$, i.e. $\iota_{m_\Psi} \iota_{m_\Phi} \phi^{tw} = 0$. Then the metric can be locally written as
    \begin{equation}
        \mathrm{d} s^2 = \dfrac{e^\mu}{F} \mathrm{d} \xi^2 \,+\, \dfrac{e^\mu}{G} \mathrm{d} \eta^2 \,+\, \dfrac{F}{e^\mu} \left(x_\xi \mathrm{d} \Psi\,+\, y_\xi \mathrm{d} \Phi \right)^2 \,+\,\dfrac{G}{e^\mu} \left(x_\eta \mathrm{d} \Psi\,+\, y_\eta \mathrm{d} \Phi\right)^2
    \end{equation}
    where $F=F(\xi)$, $G = G(\eta)$, and $x,y$ are moment maps wrt axial Killing fields $m_\Psi=\frac{\partial}{\partial \Psi},$ $m_\Phi = \frac{\partial}{\partial \Psi}$ given by
    \begin{equation}
        x = \dfrac{c_4(\eta) - c_2(\xi)}{c_1(\xi) - c_3(\eta)}\,, \qquad y = \dfrac{c_1(\xi) c_4(\eta) - c_2(\xi) c_3(\eta)}{c_1(\xi) - c_3(\eta)}\,.
    \end{equation}
    The K\"{a}hler form is given by
    \begin{equation}
        J = \mathrm{d} \left(x \mathrm{d} \Psi \,+\, y \mathrm{d} \Phi\right)\,,
    \end{equation}
    and the twistor 2-form is (up to a constant factor)
    \begin{equation}\label{eq:constr_tw-form}
        \phi^{tw} = e^{\mu} I^1 = e^{\mu} \left(x_\xi \mathrm{d} \xi \wedge \mathrm{d} \Psi + y_\xi \mathrm{d} \xi \wedge \mathrm{d} \Phi - x_\eta \mathrm{d} \eta \wedge \mathrm{d} \Psi - y_\eta \mathrm{d} \eta \wedge \mathrm{d} \Phi\right)\,.
    \end{equation}
    The pre-factor $\mu$, as well as the functions $c_{1,2,3,4}$ are constrained by the system
    \begin{equation}\label{eq:constr_lem_final-constr}
        \begin{aligned}
            x_{\xi \eta} = \mu_\xi x_\eta + \mu_\eta x_\xi\,,\\
            y_{\xi \eta} = \mu_\xi y_\eta + \mu_\eta y_\xi\,.
        \end{aligned}
    \end{equation}
   
\end{lemma}

We will proceed with solution of the final constraints in the next section. To simplify them we will use the remaining freedom in the definition of the chart. Firstly, notice that the coordinates $\xi, \eta$ are defined up to their reparameterisations. Secondly, the Killing fields $m_\Psi, m_\Phi$ are defined up to the action of constant $GL(2,\mathbb{R})$ coupled to an appropriate affine transformation of the moment maps $x,y$
\begin{equation}\label{eq:constr_affine-freedom}
    \Psi^i\, \rightarrow\,  \Lambda^i_{\,\,\,j} \Psi^j\,, \qquad x_i \,\rightarrow \, x_j \,\left(\Lambda^{-1}\right)^{j}_{\,\,\,i} + x_{(0)\,i}\,
\end{equation}
where  $\Psi^i, x_i$ collectively stand for $\Psi,\Phi$ and $x,y$, $\Lambda \in GL(2,\mathbb{R})$, and $x_{(0)\,i}$ is a constant shift of origin of moment maps. Under these transformations the functions $c_{1,2,3,4}$ change as
\begin{equation}\label{eq:constr_affine-transform-c1234}
    \begin{aligned}
        c_1 \rightarrow \dfrac{a_{21} +a_{11} c_1}{a_{22}+a_{12}c_1}\,, \qquad c_2 \rightarrow \dfrac{(a_{12} x_0 - a_{11} y_0) c_1 + \mathrm{det}\Lambda^{-1}\, c_2 + a_{22} x_0 - a_{21} y_0 }{a_{22} + a_{12} c_1}\,,\\
        c_3 \rightarrow \dfrac{a_{21} +a_{11} c_3}{a_{22}+a_{12}c_3}\,, \qquad c_4 \rightarrow \dfrac{(a_{12} x_0 - a_{11} y_0) c_3 + \mathrm{det}\Lambda^{-1}\, c_4 + a_{22} x_0 - a_{21} y_0}{a_{22} + a_{12} c_3}\,,
    \end{aligned}
\end{equation}
where $\Lambda^{-1} = \left( \begin{array}{cc} a_{11}  & a_{12} \\ a_{21}  &  a_{22}\end{array} \right)$ and $x_{(0)\, i } = (x_0, y_0)$. In particular, a following proposition follows. 
\begin{proposition}\label{prop:zero-for-c1-c3}
    In any local chart of \cref{lem:diag_chart}, one can always arrange for $c_1$ (or $c_3$) to have a zero by an appropriate rotation $\Lambda$.
\end{proposition}

Next, notice that $\xi, \eta$ enter the construction on an equal basis: indeed, instead of $I+J$ in the first step, we could have defined the coordinate $\eta$ wrt $I-J$, obtaining the same geometry. Therefore, we have an additional duality
\begin{equation}\label{eq:constr_duality}
    \xi \leftrightarrow \eta\,, \quad F \leftrightarrow G\,,\quad\, \Psi \leftrightarrow \Phi\,, \quad x \leftrightarrow y
\end{equation}
which allows for reduction of the number of special cases in the next section. 

Finally, we would like to comment that the dimensionality is crucial for this construction, especially in choosing the parameterisation of the metric at the steps 2 and 3.

\subsection{Construction of the twistor form. Special cases} 
\label{sub:construction_of_the_twistor_form_special}

The function $\mu$ can be found by solving the system \cref{eq:constr_lem_final-constr} algebraically for derivatives $\mu_{\xi}, \mu_\eta$ and checking that they are compatible
\begin{equation}\label{eq:twist-spec_int_constr}
    \mu_{\xi \eta} = \mu_{\eta \xi}\,, \qquad \mu_\xi = \frac{x_\xi y_{\xi \eta} - x_{\xi \eta} y_\xi}{x_\xi y_\eta - \xi_\eta y_\xi}\,, \qquad \mu_\eta = \frac{x_{\xi \eta} y_\eta - x_\eta y_{\xi \eta}}{x_\xi y_\eta - \xi_\eta y_\xi}
\end{equation}
It is useful to fix the freedom in $\xi,\eta$ by making two of the functions $c_{1,2,3,4}$ linear. Let us consider several important cases first. \\

\subsubsection{Product-toric geometries}

These geometries correspond to constant functions $c_{1,3}$. Let $c_{1,3} = c_{01,03} = \textrm{const}$. Notice that from \cref{eq:constr_xy-c1234_sol} the constants are non-equal $c_{01} \neq c_{03}$, and the functions $c_2, c_4$ are non-constant. We can therefore choose
\begin{equation}
    c_2 = \xi\,, \quad c_4 = \eta\,.
\end{equation}
The system \cref{eq:constr_lem_final-constr} has the solution
\begin{equation}
    \mu = \textrm{const}\,, \qquad (x, y) = \frac{1}{c_{01}-c_{03}} \left( c_{01} \eta - c_{03} \xi\,,\, \eta- \xi\right)\,.
\end{equation}

This is in fact the \textit{product-toric} geometry \cite{apostolov2003}. The standard form \cref{eq:hamiltonian_product-toric} is recovered by an affine transformation
\begin{equation}\label{eq:twist-spec_Product-toric}
    x = \xi\,, \qquad y = \eta\,, \qquad \mu = \textrm{const}\,.
\end{equation}

The twistor form $\phi^{tw}=I$ then defines a second K\"{a}hler structure on the geometry.

\subsubsection{Calabi-toric geometries}

This case corresponds to only one of the functions $c_{1,3}$ being constant. By the duality \cref{eq:constr_duality} it is sufficient to consider the case $c_1 = c_{01} = \textrm{const}$. From \cref{eq:constr_xy-c1234_sol} it follows that $c_2$ is non-constant, and it is then convenient to fix the $\xi, \eta$ coordinates by
\begin{equation}
    c_2 = \xi + c_{02}\,, \qquad c_3 = c_{01} + \frac{1}{\eta}
\end{equation}
where $c_{02}$ is a constant to be fixed later. The system \cref{eq:constr_lem_final-constr} then dictates 
\begin{equation}
    c_4 = -\frac{c_{04}}{\eta} + c_{03}\,, \qquad \mu = \mathrm{log}\, \xi\, +\, \textrm{const}\,,
\end{equation}
for some constants $c_{03}, c_{04}$, and we have fixed $c_{02}= c_{03}$. The symplectic coordinates are given by
\begin{equation}
    x= c_{03} + c_{01} c_{04} + \xi + c_{01}\, \xi \eta\,,\qquad y = c_{04} \,+\, \xi   \eta.
\end{equation}

The geometry is \textit{Calabi-toric} \cite{apostolov2003}, with the standard form \cref{eq:hamiltonian_Calabi-toric} recovered by an affine transformation
\begin{equation}\label{eq:twist-spec_Calabi-toric}
     x = \xi\,, \qquad y = \xi \eta\,, \qquad \mu = \mathrm{log}\, \xi + \textrm{const}.
\end{equation}

For Calabi geometries $(g, J, \phi^{tw} = e^{\mu} I)$, their conformal dual $(e^{-2 \mu} g, e^{-2 \mu} I, e^{-3\mu} J)$ is also Calabi, with $\tilde{\xi}= 1/\xi$ and $\widetilde{F} (\tilde{\xi}) = \tilde{\xi}^4 F(1/\tilde{\xi})$\,.

\subsubsection{Orthotoric geometries}\label{sub:orthotoric_case_3}

This case corresponds to the functions $c_{1,3}$ are non-constant, but $c_4$ diverges as $1/c_3$ at a zero of $c_3$. 
Let us set $c_1 = \xi$ and $c_3= \eta$. The compatibility constraint \cref{eq:twist-spec_int_constr} is
\begin{equation}\label{eq:twist-spec_ortho-int-constr}
    c_2'' \left( (\xi- \eta) c_4' +c_4 -c_2\right)^3 + c_4'' \left( (\xi- \eta) c_2' + c_4 - c_2\right)^3 = 0\,.
\end{equation}
It is convenient to decompose $c_4$ into a divergent $\hat{c}$ and a regular part $\tilde{c} \in C^2(M)$. Note that from \cref{eq:constr_affine-transform-c1234}, we can still use the translation of moment maps to make
\begin{equation}\label{eq:twist-spec_ortho-shift-of-origin}
    \tilde{c}(0) = \tilde{c}'(0) = 0\,.
\end{equation}

Let\footnote{We are using the asymptotic notation defined in the \cref{eq:twist-gen_asympt-notation} in the next section.} $\hat{c} = \Theta(1/c_3) = \Theta(\eta^{-1})$. The equation diverges as $\Theta (\eta^{-6})$, and the leading order is the ODE for $c_2$
\begin{equation}
    c_2'' \xi^3 + c_{01} = 0
\end{equation}
for some constant $c_{01}$ related to the behaviour of $\hat{c}$ at this order. The solution is
\begin{equation}
    c_2 = \dfrac{c_{01}}{\xi} + c_{02} + c_{03} \xi
\end{equation}
for some constants $c_{02,03}$. Substituting back into the compatibility constraint, we find
\begin{equation}
    c_4 = \dfrac{c_{01}}{\eta}\,, \qquad c_{02} = c_{03} = 0\,,
\end{equation}
since we have used the shift of origin to remove the linear part \cref{eq:twist-spec_ortho-shift-of-origin}.

The system \cref{eq:constr_lem_final-constr} then gives
\begin{equation}\label{eq:twist-spec_Orthotoric}
    x = \xi + \eta\,, \qquad y = \xi \eta\,, \qquad \mu = \mathrm{log}\, |\eta - \xi| \,+\, \textrm{const}
\end{equation}
where we have used the remaining affine freedom and coordinate inversion $(\xi, \eta) \rightarrow (1/\xi,1/\eta)$ to recover the the standard form \cref{eq:hamiltonian_Orthotoric} of the \textit{orthotoric} geometry \cite{apostolov2003}.

The conformal duals of orthotoric geometries are not orthotoric and in fact form a separate family of solutions, the \textit{conformally orthotoric} or \textit{parabolic} geometries.

\subsubsection{Conformally orthotoric (parabolic) geometries}
\label{sub:conf-orthotoric}

This case corresponds to non-constant functions $c_{1,3}$, and $c_4$ behaving as a linear function of $c_3$. We can again fix the coordinates by $c_1=\xi$, $c_3 = \eta$. The compatibility constraint \cref{eq:twist-spec_int_constr} then gives
\begin{equation}
    c_2''=0\,,
\end{equation}
i.e. $c_2$ is also a linear function. We then find that the symplectic coordinates are given by
\begin{equation}
    x = \frac{c_{00} + c_{01} \xi}{\eta - \xi}\,, \qquad y = \frac{(c_{00} + c_{01} \xi) \eta}{\eta - \xi}\,.
\end{equation}
For $c_{01} \neq 0$ the parameters $c_{00,01}$ can be removed by a coordinate inversion $(\xi,\eta) \rightarrow (1/\xi + c_{02}, 1/\eta + c_{02})$ together with another affine transformation. The full solution is then
\begin{equation}\label{eq:twist-spec_conf-ortho}
    x = \frac{1}{\eta - \xi}\,, \qquad y = \frac{\eta}{\eta - \xi}\,, \qquad \mu = - \mathrm{log}\,\left|\eta - \xi\right| \,+\, \textrm{const}\,.
\end{equation}
For a conformally orthotoric geometry $(g, J,\phi^{tw} = e^{\mu} I)$ its conformal dual $(e^{-2 \mu} g, e^{-2 \mu} I, e^{-3\mu} J)$ is orthotoric, and we have chosen the parameterisation to make it manifest.

\subsection{Construction of the twistor form. General case} 
\label{sub:construction_of_the_twistor_form_general}

In this section we will solve the remaining constraints \cref{eq:constr_lem_final-constr} in general. It is convenient to set
\begin{equation}\label{eq:twist-gen_c1_c3}
    c_1 = \xi\,, \qquad c_3= \eta\,,
\end{equation}
which simplify the compatibility constraint \cref{eq:twist-spec_int_constr} to 
\begin{equation}\label{eq:twist-gen-int-constr}
    c_2'' \left( (\xi- \eta) c_4' +c_4 -c_2\right)^3 + c_4'' \left( (\xi- \eta) c_2' + c_4 - c_2\right)^3 = 0\,.
\end{equation}

By \cref{prop:zero-for-c1-c3}, we can assume that our chart covers at least one zero of $c_3$. The process of solution splits into two different branches depending on whether $c_4$ is regular at zero up to its second derivatives or not. If the former is true, then the parameterisation we are using for the ODE will separate a further case when $c_4$ is a linear function, which is covered by \cref{sub:conf-orthotoric}.

\subsubsection{\texorpdfstring{$c_4$ is regular at a zero of $c_3$}{c4 is regular at a zero of c3}}

Now let $c_4$ be regular at zero up to the second derivative. For the moment we fix the origin of symplectic coordinates so that
\begin{equation}\label{eq:twist-gen_c4-at-zero}
    c_4|_0 = c_4'|_0 = 0\,.
\end{equation}
The leading order of the constraint \cref{eq:twist-gen-int-constr} at zero is then
\begin{equation}
    c_2'' c_2^3 + c_{04} \left(c_2 - \xi\, c_2'\right)^3 = 0
\end{equation}

where we have defined $c_4 = \frac{1}{2} c_{04} \eta^2 + o(\eta^2)$. If $c_{04} = 0$, we have $c_2$ a linear function of $\xi$, and substituting it back into \cref{eq:twist-gen-int-constr} gives us $c_2'' = c_4''=0$\,, i.e. the case of \cref{sub:conf-orthotoric}.

Let us then assume that $c_{04}$ is non-zero, and we can introduce a new function $f(\xi)$ by $c_2:= c_{04} f$\,, so that the ODE simplifies to
\begin{equation}\label{eq:twsit-gen_ODE-reg}
    f'' f^3 + (f - \xi f')^3 = 0\,.
\end{equation}

This equation is solved in the \cref{app:solution_of_the_ODE-reg}. Its solutions are
\begin{gather}\label{eq:twsit_gen_c2-sol1}
    c_2 = \hat{c}_{04} \sqrt{\xi^2 + c_{01}^2} +\, \textrm{linear terms}\,,\\\label{eq:twsit_gen_c2-sol2}
    c_2 = \hat{c}_{04} \sqrt{|c_{01} \xi^2 + 2 c_{02} \xi + c_{02}^2|} +\, \textrm{linear terms}\,
\end{gather}
where $c_{01,02}$ and  $\hat{c}_{04}$ are constants, and $c_{01} \neq 0$ by regularity. Shifting and rescaling $(\xi, \eta)$, and using the affine transformations again, the two branches can be written collectively as  
\begin{equation}\label{eq:twsit_gen_c2-sol}
    c_1 = \xi\, \qquad c_2 = \sqrt{|\xi^2 + \alpha|}\,, \qquad c_3 = \eta\,, \qquad \alpha = \{0, \pm 1\}\,.
\end{equation}
Substituting back into \cref{eq:twist-gen-int-constr}, we then find that the solution has two branches
\begin{gather}\label{eq:twsit_gen_c24-sol}
    c_1 = \xi\,, \qquad c_2 = \sqrt{|\xi^2 + \alpha|}\,, \qquad c_3 = \eta\,, \qquad (c_4)_{\pm} = \pm\sqrt{|\eta^2 +\alpha|}\,.
\end{gather}

In the full solution the two branches persist
\begin{gather}\label{eq:twist-gen-sol-alpha}
    (x,y)_{\pm} = \frac{1}{\eta - \xi} (|R(\xi)|^{1/2} \pm |R(\eta)|^{1/2},\ \eta \,|R(\xi)|^{1/2} \pm \xi\,|R(\eta)|^{1/2})\,,\\\nonumber
    \mu_{\pm} = \mathrm{log}\,\left|\dfrac{|R(\xi)|^{1/2} |R(\eta)|^{1/2} \pm R(\xi,\eta)}{\eta - \xi} \right| \,+\, \textrm{const}\,,
\end{gather}
where $R(\xi,\eta) = \xi \eta + \alpha$ is a polarisation of a quadratic polynomial $R(z) = z^2 + \alpha$. The geometries break into \textit{hyperbolic}, \textit{parabolic} and \textit{elliptic} classes depending if $R(z)$ has two, one or zero real roots. For parabolic family $\alpha=0$ the two branches are isomorphic to the conformally orthotoric family \cref{eq:twist-spec_conf-ortho}. In the elliptic case $\alpha=1$, the two branches are locally isomorphic and can be collectively written as
\begin{equation}\label{eq:twist-gen-sol-elliptic}
    x = \frac{\eta + \xi}{\eta- \xi}\,, \qquad y = \frac{\eta \xi - 1}{\eta - \xi}\,, \qquad \mu = \mathrm{log}\,\left|\frac{\eta \xi + 1}{\eta - \xi}\right| \,+\, \textrm{const}
\end{equation}
where the transformation to $(\xi,\eta)_{\pm}$ coordinates of two branches of \cref{eq:twist-gen-sol-alpha} can be found from
\begin{equation}
    (\xi \, \xi_{\pm} + 1)^2 = \xi^2_{\pm} + 1\,, \qquad (\eta \, \eta_{\pm} + 1)^2 = \eta^2_{\pm} + 1\,.
\end{equation}
The conformal dual of elliptic geometries is also elliptic with $(\tilde{\xi}, \tilde{\eta})= (-\xi, 1/\eta)$ and reparameterisation $\tilde{F}(\tilde{\xi}) = F( - \tilde{\xi})$ and $\tilde{G}(\tilde{\eta}) = \tilde{\eta}^4 G(1/\tilde{\eta}).$

The hyperbolic branches $\alpha=-1$ are also locally isomorphic and can be conveniently expressed as
\begin{equation}\label{eq:twist-gen-sol-hyperbolic}
       x = \frac{1}{\eta - \xi}\,, \qquad y = \frac{\xi \eta}{\eta - \xi}\,, \qquad \mu = \mathrm{log}\,\left|\frac{\eta + \xi}{\eta - \xi}\right| \,+\, \textrm{const}\,.
   \end{equation}
where where the transformation to $(\xi,\eta)_{\pm}$ coordinates of two branches of \cref{eq:twist-gen-sol-alpha} can be found from
\begin{equation}
    \xi^2 (\xi_{\pm}+1)^2 = |\xi_{\pm}^2 - 1|\,, \qquad \eta^2 (\eta_{\pm}+1)^2 = |\eta_{\pm}^2 - 1|\,.
\end{equation}
The conformal dual of hyperbolic geometries is also hyperbolic with the transformation given by inversion of sign of $\xi$.

This completes the list of independent 4d toric K\"{a}hler geometries with a twistor 2-form. In the following subsection we will recover the known families again.

\subsubsection{\texorpdfstring{$c_4$ is irregular at a zero of $c_3$}{c4 is irregular at at zero of c3}}

We can decompose $c_4 = \hat{c} + \tilde{c}$ where $\hat{c}$ is divergent and $\tilde{c}$ is $C^2$ at a zero of $c_3$. Again from \cref{eq:constr_affine-transform-c1234}, we use the shift of origin of symplectic coordinates $x,y$ to set 
\begin{equation}\label{eq:twist-gen_irreg-choice-of-origin}
    \tilde{c}(0) = \tilde{c}'(0) = 0\,.
\end{equation}

Let us study the asymptotic behaviour of the equation \cref{eq:twist-gen-int-constr} at zero. For convenience, we introduce a standard asymptotic notation
\begin{gather}\label{eq:twist-gen_asympt-notation}
    f = o(g)\, \textrm{ if }\, \lim\limits_{\eta \to 0} f/g = 0\,, \qquad f = \omega(g)\, \textrm{ if }\, \lim\limits_{\eta \to 0} g/f = 0\,,\\ 
    \nonumber f = \Theta(g)\, \textrm{ if }\, 0< \left|\lim\limits_{\eta \to 0} f/g \right| < \infty \,.
\end{gather}

\begin{itemize}
    \item Let $\hat{c} = \omega(\eta^{-1})$. Then the leading divergent order of equation is
    \begin{equation}
        \hat{c}''\, \hat{c}^3 = 0
    \end{equation}
    which implies that such divergences vanish.  

    \item For $\hat{c} = \Theta(\eta^{-1})$, we have the orthotoric case, see \cref{sub:orthotoric_case_3}.

    \item For $\hat{c} = o(\eta^{-1})$, $\hat{c} = \omega(\eta^{1/2})$, the leading divergence is
    \begin{equation}
        c_2''\, \xi^3 \hat{c}'^3 = 0\,.
    \end{equation}
    The $c_2$ is therefore a linear function, and we recover the parabolic case.

    \item Let $\hat{c} = o(\eta^{1/2})$. Then the only divergent term is
    \begin{equation}
        c_4'' \left(\xi c_2' - c_2\right)^3 = 0
    \end{equation}
    which again implies that $c_2$ is a linear function as in the previous case.

    \item Finally, let $\hat{c} = \Theta (\eta^{1/2})$. We can parameterise it as $\hat{c} = \frac{1}{c_{01}} \eta^{1/2} + \hat{c}_{o}$ where $c_{01}$ is a non-zero constant and $\hat{c}_{o} = o(\eta^{1/2})$. The equation diverges as $\Theta (\eta^{-3/2})$, and the leading order is the ODE for $c_2$
    \begin{equation}\label{eq:twsit-gen_ODE-irreg}
        c_2'' \xi^3 - 2 c_{01}^2\, (\xi c_2' - c_2)^3 = 0\,.
    \end{equation}
    This equation is solved in \cref{app:solution_of_the_ODE-irreg}, and its relevant solutions are
    \begin{gather}\label{eq:twist-gen-irreg_c2-sol1}
        c_{2} = \pm \frac{\sqrt{\xi}}{c_{01}} \left[ 1 + \left(c_{03}^{-1}+\frac{1}{4} c_{02}^2 c_{03}\right) \xi + c_{02} \sqrt{\xi (\xi + c_{03})} \right]^{1/2}, \quad c_{03} \neq 0
        \,,\\\label{eq:twist-gen-irreg_c2-sol2}
        c_{2} = \pm \frac{\sqrt{\xi}}{c_{01}} \left[ 1 - \left(c_{03}^{-1}-\frac{1}{4} c_{02}^2 c_{03}\right) \xi + c_{02} \sqrt{\xi (c_{03}-\xi)} \right]^{1/2}, \quad c_{03} > 0
        \,,\\
        \label{eq:twist-gen-irreg_c2-sol3}
        c_{2} = \pm \frac{1}{c_{01}} \sqrt{\xi (1 + c_{02} \xi)}\,.
    \end{gather}
    The function $c_4$ can now be obtained by substituting these solutions into \cref{eq:twist-gen-int-constr} and expanding it in the powers of $\xi$ at zero. We find that these solutions are isomorphic to the above. First of all using our choice of origin \cref{eq:twist-gen_irreg-choice-of-origin}, solutions \cref{eq:twist-gen-irreg_c2-sol1,eq:twist-gen-irreg_c2-sol2} reduce to the last one \cref{eq:twist-gen-irreg_c2-sol3}
    \begin{equation}
        c_{2} = \pm \frac{1}{c_{01}} \sqrt{\xi (1 + c_{03}^{-1}\, \xi)}\,, \qquad c_{4} = \pm \frac{1}{c_{01}} \sqrt{\eta (1 + c_{03}^{-1}\, \eta)}\,, \qquad c_{02} = 0\,,\ c_{03} \neq 0\,.
    \end{equation}
   For the last branch \cref{eq:twist-gen-irreg_c2-sol3} we find $c_{4} = \frac{1}{c_{01}} \sqrt{\eta (1 + c_{02} \eta)}$, i.e. for $c_{02} \neq 0$ the solution is isomorphic to the cases of the previous section. The case $c_{02} = 0$ corresponds to the hyperbolic branch, and after a coordinate transformation $(\xi, \eta) \rightarrow (\xi^2,\eta^2)$ it takes the form \cref{eq:twist-gen-sol-hyperbolic}. 

\end{itemize}

\subsection{Non-isomorphism of the families}
\label{sub:non-isomorphism_of_the_families}

In the previous subsection we have completed the list of families of toric K\"{a}hler geometries with a SD twistor 2-form in the negative eigenspace of $\mathrm{i}_\varphi$; the positive eigenspace is incompatible with smooth axis by \cref{prop:no-I2}, and to finish the proof of \cref{thm:main} we only need to prove that the families are independent up to \textit{family isomorphisms}, that is up to diffeomorphisms and a change of metric functions $(F,G) \rightarrow (\tilde{F},\tilde{G})$. 

Firstly, three of these families, the product-toric, the Calabi-toric and the orthotoric, are well-known to the literature and together comprise a class of all toric K\"{a}hler geometries with a hamiltonian 2-form; their non-isomorphism follows from the construction \cite{apostolov2003}. In \cref{lem:ham_curv-constr} of the next section, we show that the existence of a hamiltonian 2-form is equivalent to a constraint on the Ricci form whose SD part must be a multiple of a twistor 2-form. We can check this condition for each geometry in the given basis of SD forms \cref{eq:constr_SD_basis}. Firstly, for any toric geometry, the contraction of $\mathcal{R}$ on $I^2$ is zero by the symmetry. For new geometries we then find
\begin{equation}
    \frac{1}{4} I^{3\,\, mn} \mathcal{R}_{m n} = \dfrac{2 \sqrt{F G}}{\eta - \xi} \not\equiv 0
\end{equation}
for hyperbolic, elliptic and parabolic (conformally orthotoric) geometries in the parameterisation \cref{eq:twist-gen-sol-alpha}\,. Notice in particular, that the orthotoric geometry is not isomorphic to its conformal dual, which happens for all other families.
  
To show that the new geometries are mutually distinct, we consider how a twistor 2-form relates to symmetries. In a similar setting, a hamiltonian 2-form generates two hamiltonian potentials for the Killing vector fields (which can happen to be linearly dependent, or, in a product-toric case, zero). For twistor 2-forms we make a following observation.
\begin{proposition}\label{prop:square_norm is quadratic}
    For 4d toric K\"{a}hler geometries, the square norm $e^{2 \mu}$ of the twistor 2-form is at most quadratic in moment maps.
\end{proposition}
\begin{proof}
    This can be checked directly. Using explicit expressions \cref{eq:twist-spec_Product-toric,eq:twist-spec_Calabi-toric,eq:twist-spec_Orthotoric,eq:twist-spec_conf-ortho,eq:twist-gen-sol-elliptic,eq:twist-gen-sol-hyperbolic} for our geometries we find
    \begin{center}
        \begin{tabular}{ |l|l||l|l| }
        \hline
            Family        &  $e^{2 \mu}$    & Family        &   $e^{2 \mu}$\\
        \hline
            Product-toric &  const          & Elliptic      &   $|x^2+y^2 - 1|$\\ 
            Calabi-toric  &  $x^2$          & Parabolic     &   $x^2 $ \\  
            Orthotoric    &  $|x^2 - 4 y|$  & Hyperbolic    &   $|4 x y + 1|$\\
        \hline  
        \end{tabular}
    \end{center}   
\end{proof}
Since there is no such an affine transformation of moment maps $x,y$ that can bring the square norm of, e.g., elliptic family to parabolic or hyperbolic form, we deduce that the families are indeed independent. This concludes the proof of \cref{thm:main}.

Some comments are in order. Firstly, we have observed the \cref{prop:square_norm is quadratic} rather than derived it. At the same time, the simple algebraic form of the square norm in terms of moment maps suggests that this result can be more naturally obtained using the tools of algebraic geometry or, perhaps, conformal geometry \cite{Dunajski:2024uhh}. Secondly, since toric geometries form a sufficiently broad class, one can expect that this form holds generally for all K\"{a}hler geometries with a twistor 2-form. We therefore put forward \cref{conj:conj}, mostly in the hope that it will help in extending our results to a more general (and rarely studied) case of single $U(1)$ symmetry. Furthermore, $SU(2)$-symmetric geometries can also provide an interesting testing ground for the validity of our proposal.

\section{Hamiltonian 2-forms on toric geometries} 
\label{sec:hamiltonian_2_forms}

Hamiltonian 2-forms have been introduced in \cite{apostolov2003} as a special construction based on twistor 2-forms.
\begin{definition}\label{def:hamiltonian-2form}
    A hamiltonian 2-form on a K\"{a}hler surface $(M, g, J)$ is a closed J-invariant 2-form $\phi^{ham}$ whose self-dual part $\phi^{+} \not\equiv 0$ is a twistor 2-form.  
\end{definition}
We added the assumption of non-vanishing SD part to exclude the K\"{a}hler form itself and its multiples. This definition can be equivalently rewritten as a differential equation\footnote{In coordinate notation, and writing $\sigma_m = \partial_m \sigma$, it takes the form
\begin{equation}
    \nabla_m \phi^{ham}_{np} = \frac{2}{3} \sigma_m J_{np} - \frac{2}{3} J_{m[n} \sigma_{p]} -\frac{2}{3} g_{m[n} J_{p] q} \sigma^q\,.
\end{equation}
} $\forall X \in TM$
\begin{equation}\label{eq:hamiltonian_Ham-def}
    \nabla_X \phi^{ham} = \frac{2}{3} X(\sigma) \ J -\frac{1}{3} \mathrm{d} \sigma \wedge JX - \frac{1}{3} X \wedge J \mathrm{d} \sigma
\end{equation}   
where $\sigma = \frac{1}{4} \phi^{ham}_{mn} J^{mn}$, $(JX)_m = J_{mn} X^n$. Not all twistor forms give rise to hamiltonian 2-forms, and we study necessary and sufficient conditions in \cref{sub:curvature_constraints_hamiltonian}.

\subsection{Product-toric, Calabi-type and orthotoric geometries}

To justify the name, hamiltonian 2-forms generate hamiltonian potentials for commuting Killing vector fields. This can be summarised in the following proposition \cite{apostolov2003}.
\begin{proposition}
    Let $(M,g,J)$ be a K\"{a}hler surface and let $\phi^{ham} = \phi^{tw} + \sigma J$ be a hamiltonian 2-form.\\
    Then the trace $\sigma$ and the pfaffian\footnote{Recall that the \textit{pfaffian} pf($\psi$) of a 2-form $\psi$ is defined by $\textrm{pf}(\psi) = 2 \star (\psi \wedge \psi).$} $\pi$ of the \textit{normalised} hamiltonian 2-form $\phi^{norm} = \frac{1}{2} \phi^{tw} + \frac{1}{6} \sigma J$ are either constants or hamiltonian potentials for commuting Killing vector fields $K_1 = J \mathrm{d} \sigma$ and $K_2 = J\mathrm{d} \pi$.
\end{proposition}
This means that geometries with a hamiltonian 2-form fall into three distinct families:
\begin{enumerate}
    \item Bik\"{a}hler: $K_1 = K_2 = 0$\,,
    \item Calabi-type: $K_1$ is non-vanishing identically, but $K_1 \wedge K_2 = 0$\,,
    \item Orthotoric: $K_1 \wedge K_2 \neq 0$ on an open dense set.
\end{enumerate}
For toric geometries one can further show that bik\"{a}hler geometries necessarily take the form of a Cartesian product of two axisymmetric Riemannian surfaces (hence the ``product-toric'' name). The standard form for such geometries is given by
\begin{equation}\label{eq:hamiltonian_product-toric}
    \begin{aligned}
        h &= \frac{\mathrm{d}\xi^2}{F(\xi)}\,+\,\frac{\mathrm{d}\eta^2}{G(\eta)}\,+\,F(\xi)\mathrm{d}\Psi^2\,+\,G(\eta) \mathrm{d}\Phi^2\,,\\
        J &= \mathrm{d} \left(\xi \,\mathrm{d} \Psi \,+\, \eta\, \mathrm{d}\Phi\right)\,.
    \end{aligned}
\end{equation}
The examples of such geometries include Euclidean space which in terms of symplectic coordinates is given by the potential\footnote{A familiar form can be recovered by introducing radii $r_i$, $x_i := r_i^2$.} $\mathrm{g} = \sum_i \frac{1}{2} x_i \mathrm{log}\, x_i$. The product-toric chart is then
\begin{equation}
     F_{\mathbb{R}^4}(\xi)=2 \xi\,, \qquad G_{\mathbb{R}^4}(\eta)=2 \eta\,, \qquad x_1 = \xi, \qquad x_2 = \eta\,.
\end{equation}
Euclidean space is the only geometry which possesses different hamiltonian 2-forms of all three types, and other charts are given below. 

Calabi-type geometries are named after the Calabi construction \cite{calabi2016} of metrics on the total space of a Hermitian line bundle over a Riemann surface. Given the standard chart for toric Calabi-type geometries
\begin{equation}\label{eq:hamiltonian_Calabi-toric}
    \begin{aligned}
        h &= \frac{\xi}{F(\xi)}\mathrm{d}\xi^2\,+\,\frac{\xi}{G(\eta)}\mathrm{d}\eta^2\,+\,\frac{F(\xi)}{\xi}\left(\mathrm{d}\Psi + \eta \mathrm{d} \Phi\right)^2\,+\, \xi \,G(\eta) \mathrm{d}\Phi^2\,,\\
        J &= \mathrm{d} \left(\xi \,\mathrm{d} \Psi \,+\, \xi \eta\, \mathrm{d}\Phi\right)\,,
    \end{aligned}
\end{equation}
the Riemann surface $\Sigma$ is the $(\eta, \Phi)$ two-dimensional subspace, and the Killing vector field $K = \partial / \partial \Psi$ generates the natural $U(1)$ action on this line bundle. The examples of geometries in this class include all $SU(2) \times U(1)$-invariant K\"{a}hler surfaces \cite{apostolov2003}, and, hence, the Euclidean space and the $\mathbb{C}P^2$. The former is given by
\begin{equation}
    F_{\mathbb{R}^4}(\xi) = \xi^2\,, \qquad G_{\mathbb{R}^4}(\eta) = 1- \eta^2\,, \qquad x_1 =\xi (1- \eta)\,, \qquad x_2 = \xi (1+\eta)\,.
\end{equation}
The standard metric for the latter can be found from the symplectic potential $\mathrm{g} = \sum_i \frac{1}{2} x_i \mathrm{log}\, x_i - (x_1 + x_2 - 1) \mathrm{log}\, (1 - x_1 - x_2)$, and in Calabi-type chart it takes the form 
\begin{equation}
    F_{\mathbb{C} P^2}(\xi) = 2 \xi (1- \xi^2)\,, \qquad G_{\mathbb{C} P^2}(\eta) = 2 (1- \eta^2)\,, \qquad x_1 =\frac{1}{2} \xi (1-\eta)\,, \qquad x_2 = \frac{1}{2} \xi (1+\eta)\,.
\end{equation}

The standard form of orthotoric geometries is
\begin{equation}\label{eq:hamiltonian_Orthotoric}
    \begin{aligned}
        h &= \frac{\xi - \eta}{F(\xi)}\mathrm{d}\xi^2\,+\,\frac{\xi - \eta}{G(\eta)}\mathrm{d}\eta^2\,+\,\frac{F(\xi)}{\xi - \eta}\left(\mathrm{d}\Psi + \eta \mathrm{d} \Phi\right)^2\,+\, \frac{G(\eta)}{\xi - \eta}\left(\mathrm{d}\Psi + \xi \mathrm{d} \Phi\right)^2\,,\\
        J &= \mathrm{d} \left( (\xi + \eta) \,\mathrm{d} \Psi \,+\, \xi \eta\, \mathrm{d}\Phi\right)\,.
    \end{aligned}
\end{equation}
The Euclidean space in this chart is described by
\begin{equation}
    F_{\mathbb{R}^4}(\xi) = \xi^2 - \xi\,, \qquad G_{\mathbb{R}^4}(\eta) = \eta - \eta^2\,, \qquad x_1 = 2 (\xi -1) (1- \eta)\,, \qquad x_2 = 2\xi \eta\,.
\end{equation}
The complex projective space $\mathbb{C}P^2$ also possesses two independent hamiltonian 2-form structures. In the orthotoric chart it reads as
\begin{equation}
    \begin{aligned}
        F_{\mathbb{C} P^2}(\xi) = 2 \xi (\xi^2 - 1)\,&, \qquad G_{\mathbb{C} P^2}(\eta) = 2 \eta (1- \eta^2)\,,\\
        x_1 =\frac{1}{2} (\xi + 1) (1 + \eta)\,&, \qquad x_2 = \frac{1}{2} (\xi - 1) (1 - \eta)\,.
    \end{aligned}
\end{equation}

Finally, the symplectic potentials for product-toric, toric Calabi-type and orthotoric geometries were found in \cite{Lucietti:2023mvj}.



\subsection{Curvature constraints on existence of the Hamiltonian 2-forms} 
\label{sub:curvature_constraints_hamiltonian}

Hamiltonian 2-forms pose the following constraint on the curvature
\begin{proposition}\label{prop:ham_curv-constr-Apostolov}
    If $(M,g,J)$ admits a hamiltonian 2-form then the SD part of the Ricci tensor is proportional to twistor 2-form.
\end{proposition}
This proposition follows as a simple corollary from Proposition 4 in \cite{apostolov2004}. We discover that there is a partial converse at least for toric geometries.
\setcounter{lemma}{0}
\begin{lemma}\label{lem:ham_curv-constr}
    Let $(M,g,J)$ be a 4d toric K\"{a}hler geometry with a SD twistor 2-form $\phi^{tw}$ that is invariant under the torus symmetry and belongs to the negative eigenspace of the inversion of angles $\mathrm{i}_\varphi$, i.e. $\iota_{m_1} \iota_{m_2} \phi^{tw} = 0$.\\ Then $(M,g,J)$ admits a hamiltonian 2-form $\phi = \phi^{tw} + \sigma J$ for some $\sigma$ if and only if the SD part of the Ricci form is a multiple of the twistor 2-form
    \begin{equation}
        \phi^{tw} \propto \mathcal{R}^+ =\mathcal{R} - \frac{R}{4} J\,.
    \end{equation}
\end{lemma}
\begin{proof}
    We only need to prove the ``if'' direction. As usual, write the twistor form as $\phi^{tw}= e^{\mu} I$ where $I$ is a unit SD form and $e^\mu$ its norm. Notice that from \cref{cor:alt-form-of-tw-eqn} for any orthonormal basis $(I^1=I,I^2,I^3)$ one deduces the connection on the bundle of SD forms
    \begin{gather}
            \label{eq:hamiltonian_sd-deriv-I2}
        \nabla I^2 = \alpha \otimes I^3 \,-\, I^3 \mathrm{d}\mu \otimes I^1\,,\\
            \label{eq:hamiltonian_sd-deriv-I3}
        \nabla I^3 = I^2 \mathrm{d}\mu \otimes I^1 \,-\, \alpha \otimes I^2\,\hphantom{,}
    \end{gather}
    for some 1-form $\alpha$ which depends on a Levi-Civita connection.  

    The closure condition for the Hamiltonian 2-form
    \begin{equation}
        \mathrm{d} (\phi^{tw} + \sigma J) = 0 
    \end{equation}
    can be always solved for $\sigma$ 
    \begin{equation}
        \mathrm{d} \sigma = - 3 (I J) \mathrm{d} e^{\mu}
    \end{equation}
    provided that the RHS is closed, which is the integrability condition
    \begin{equation}\label{eq:hamiltonian_ham-int-cond}
        \mathrm{d} \left(\iota_{\mathrm{d} e^\mu} I J\right) = 0
    \end{equation} 
    To simplify it further we will assume the symplectic chart $(x_i, \varphi^i)$ where $x_i$ are moment maps wrt $m_i = \partial_{\varphi^i}$ Killing fields. By our assumption $\mathcal{L}_{m_i} \phi^{tw} = 0$. 

    Introduce an orthonormal basis on $\Lambda^2 M$ as $\{J^1=J, J^2,J^3, I^1 = I, I^2, I^3 = I I^2\}$ where the former are anti-self-dual and the latter are self-dual 2-forms, and the positive eigenspace of $\mathrm{i}_\varphi$ is spanned by $J^2, I^2$.

    From the toric invariance of a twistor form, the integrability constraint is then a 2-form with a single component proportional to $\mathrm{d} x_1 \wedge \mathrm{d} x_2 \propto (I^2 + J^2)$, and, consequently, it is sufficient to show that contraction of \cref{eq:hamiltonian_ham-int-cond} on $I^2$ vanishes. We find that
    \begin{equation}\label{eq:hamiltonian_ham-cond-2}
        \frac{1}{2} \textrm{tr}\left(I^2 \rightarrow \mathrm{d} \left(\iota_{\mathrm{d} e^\mu} I J\right) \right) =  I^{2\,\, mn} J^{\phantom{m} p}_{m} \nabla_n \left(I_{pq} \nabla^q e^\mu\right) = \left(J I^3\right)^{mn} \nabla_m \nabla_n e^\mu.
    \end{equation}

    Now consider the twistor form equation \cref{eq:corollary-nabla_I}. Acting with a second covariant derivative it follows that
    \begin{equation}
        2 J^{mn} I^{2\,\,pq} \nabla_m \nabla_n I_{pq} = 4 \mathcal{R}_{pq} I^{3\,\,pq} = 0
    \end{equation}
    where the last equality follows from our assumption. On the other hand, using \cref{eq:corollary-nabla_I,eq:hamiltonian_sd-deriv-I2,eq:hamiltonian_sd-deriv-I3}
    \begin{equation}
        2 J^{mn} I^{2\,\,pq} \nabla_m \nabla_n I_{pq} = 8 (J I^3)^{mn} \left(\nabla_m \nabla_n \mu + \nabla_m \mu \nabla_n \mu \right) = 8 e^{-\mu} (J I^3)^{mn} \nabla_m \nabla_n e^\mu
    \end{equation}
    which is exactly \cref{eq:hamiltonian_ham-cond-2}.    
\end{proof}

This proposition can be also seen as a generalisation of \cite[Lemma 4]{apostolov2003}, which was proven for weakly self-dual surfaces, a class of geometries whose Ricci form is the hamiltonian 2-form.

\section*{Acknowledgements}
The author is grateful to Maciej Dunajski, James Lucietti, Dmitri Bykov and Christopher Couzens for some helpful discussions. SO acknowledges support by the ``BASIS" Foundation Leader Grant 24-1-1-82-5.

\section*{Statements and declarations}
\textbf{Data Availability} \ Data sharing not applicable to this article as no datasets were generated or analysed during the current study.\\
\textbf{Competing Interests:} The author has no Conflict of interest to declare that is relevant to the content of this article.
\newpage
\appendix

\section{Ricci curvature for geometries with twistor 2-form} 
\label{app:ricci_curvature_for_geometries_with_twistor_2_form}


For toric geometries the Ricci curvature is most conveniently found from the determinant of Gram matrix of Killing fields of toric symmetry.
\begin{proposition}
    Let $(M,g,J)$ be a toric K\"{a}hler geometry. Then the Ricci form and the scalar curvature are expressed through the Gram determinant of the toric Killing vector fields $G_{gram}:=\mathrm{det}\, g(m_i, m_j)$ as 
    \begin{equation}
        \mathcal{R} = \mathrm{d} P\,, \qquad P = \frac{1}{2} J \mathrm{d} (\mathrm{log}\, G_{gr})\,, \qquad R = - \Delta (\mathrm{log}\,  G_{gr})\,.
    \end{equation}
\end{proposition}
\begin{proof}
    This is most clearly seen in the complex chart \cite{abreu2001}:
    \begin{align}\nonumber
    &g = F_{ij} \mathrm{d}u^i \mathrm{d}u^j \,+\, F_{ij} \mathrm{d}\varphi^i\,\mathrm{d}\varphi^{j}\,,\\\label{eq:app_kahl_geometry-in-cx-coords}
    &F_{ij}(u) = \frac{\partial}{\partial u^i} \frac{\partial}{\partial u^i} \mathrm{f}\,,\\\nonumber
    &J(\mathrm{d} \varphi^i , \frac{\partial}{\partial u^j}) = - J(\mathrm{d} u^i , \frac{\partial}{\partial \varphi^j}) = \delta^i_j
\end{align}
where $\mathrm{f}$ is the K\"{a}hler potential and $m_i = \frac{\partial}{\partial \varphi^i}$ are toric Killing fields. In this chart the determinant of the metric is the square of the determinant of the Gram matrix
\begin{equation*}
    \mathrm{det} g = (\mathrm{det} F_{ij})^2\,,
\end{equation*}
Using standard definitions $\mathcal{R} = - i \partial \bar{\partial} \mathrm{log}\, g$ and $R = \mathcal{R}_{mn} J^{mn}$, we deduce the result.    
\end{proof}

We now provide the values of Gram determinant for all geometries
\begin{center}
        \begin{tabular}{ |l|l||l|l| }
        \hline
            Family        &  $G_{gr}^{1/2}$    & Family        &   $G_{gr}^{1/2}$\\
        \hline
            Product-toric &  $(FG)^{1/2}$      & Elliptic      &   $2\dfrac{(FG)^{1/2}}{(\eta - \xi)^2}$\\[0.1pt]
                          &                    &               & \\
            Calabi-toric  &  $(FG)^{1/2}$      & Parabolic     &   $\dfrac{(FG)^{1/2}}{(\eta - \xi)^2}$ \\[0.1pt]
                          &                    &               & \\
            Orthotoric    &  $(FG)^{1/2}$      & Hyperbolic    &   $\dfrac{(FG)^{1/2}}{(\eta - \xi)^2}$\\
        \hline  
        \end{tabular}
    \end{center}

The scalar curvature can be found by taking the laplacian, and the determinant of the metric in \cref{eq:constr_metric} chart can be found from the above table by
\begin{equation}
    \left(\mathrm{det} \,g\right)^{1/2}  = e^{\mu} (FG)^{-1/2}\,G_{gr}^{1/2}\,.
\end{equation}
For three hamiltonian 2-form families (product-toric, Calabi-toric and orthotoric) the scalar curvature can be collectively written as
\begin{equation}
    R = - e^{-\mu} (F'' + G'')\,,
\end{equation}
while for the elliptic, parabolic and hyperbolic geometries the collective formula stands as
\begin{equation}
    R = - e^{-\mu} (F'' + G'') - \frac{6 e^{-\mu}}{\eta - \xi} (F' - G') - \frac{12 e^{-\mu}}{(\eta - \xi)^2} (F + G)\,.
\end{equation}

\section{Solutions of the ODEs}
\label{app:solution_of_the_ODEs}

\subsection{\texorpdfstring{Solution of the ODE \cref{eq:twsit-gen_ODE-reg}}{Solution of the first ODE}} 
\label{app:solution_of_the_ODE-reg}

The equation \ $f'' f^3 + (f - \xi f')^3 = 0$ is generalised homogeneous of degree $2$, that is, it is invariant under a symmetry
\begin{equation}
    \xi \rightarrow t \xi\,, \qquad f \rightarrow t^2 f\,, \qquad f^{(n)} \rightarrow t^{2-n} f^{(n)}\,.
\end{equation}
Such a symmetry allows the reduction of order by a following procedure. Introduce a substitution
\begin{equation}
    f = z(t) e^{2 t} \qquad \textrm{where }\ \xi=e^t\ (\xi>0) \ \textrm{ or } \ \xi=-e^t\ (\xi<0)\,.
\end{equation}
The derivatives then transform as
\begin{equation}
    f'(\xi) = (z' + 2 z)e^{t}\,, \qquad f''(\xi) = (z'' + 3 z' + 2 z)\,.
\end{equation}
Under this substitution the ODE becomes autonomous
\begin{equation}
    z^3 \left(z''+3 z'-1\right) -z'^3 -3 z'^2 z -3 z' z^2 +2 z^4 = 0\,,
\end{equation}
and can be further reduced by the substitution $p(z):= z'$ to a first order ODE
\begin{equation}
    (p'p + 3p +2 z)\,z^3 -(p+z)^3=0\,.
\end{equation}
This equation can be solved explicitly
\begin{equation}\label{eq:app-psol-reg}
    p = \dfrac{z(2z-1)}{1+z\,\left( \dfrac{\epsilon}{\sqrt{1+ (2z-1) \hat{c}_{01}}} -1 \right)}
\end{equation}
where $\hat{c}_{01}$ is an integration constant and $\epsilon = \{0, \pm 1\}$\,. The function $z$ is then found from solving $p(z)=z'$. We find that $\epsilon = \pm 1$ branches merge, and
\begin{equation}
    f = z\, e^{2 t} = \hat{c}_{02} \left(\xi + \hat{c}_{01} \hat{c}_{02} \,\pm\, \sqrt{(1-\hat{c}_{01}) \xi^2 + 2 \hat{c}_{01} \hat{c}_{02}\, \xi + \hat{c}_{01}^2 \hat{c}_{02}^2}\right)\,.
\end{equation}
The case $\hat{c}_{01} = 0$ would correspond to $f$, hence, $c_{2}$ being a linear function, which is the linear solution of \cref{sub:conf-orthotoric}. We therefore assume otherwise, and introduce
\begin{equation}
    c_{01} = 1 - \hat{c}_{01}\,, \qquad c_{02} = \hat{c}_{01} \hat{c}_{02}\,, \qquad \hat{c}_{04} = c_{04}\,\hat{c}_{02}
\end{equation}
to match the notation of \cref{eq:twsit_gen_c2-sol2}. Finally, for the branch $\epsilon = 0$ we find
\begin{equation}
    f = z\, e^{2 t} = - \alpha^2 \left( 1 \pm\, \sqrt{\xi^2 + \alpha^2}\right)
\end{equation}
where $\alpha:= e^{\hat{c}_{01}}$\,. To match the notation of \cref{eq:twsit_gen_c2-sol1}, define $\hat{c}_{04} = - c_{04} \alpha^2$\,.

\subsection{\texorpdfstring{Solution of the ODE \cref{eq:twsit-gen_ODE-irreg}}{Solution of the second ODE}} 
\label{app:solution_of_the_ODE-irreg}

For brevity denote $f:= c_{2}$. The ODE is then recast as
\begin{equation}
    f'' \xi^3 -2 c_{01}^2  (\xi f' - f)^3 = 0\,.
\end{equation}
It is invariant under a generalised homogeneous symmetry of degree $1/2$
\begin{equation}
    \xi \rightarrow t \xi\,, \qquad f \rightarrow t^{1/2} f\,, \qquad f^{(n)} \rightarrow t^{1/2-n} f^{(n)}\,.
\end{equation}
Introduce a substitution
\begin{equation}
    f = z(t) e^{t/2} \qquad \textrm{where }\ \xi=e^t\ (\xi>0) \ \textrm{ or } \ \xi=-e^t\ (\xi<0)\,.
\end{equation}
The derivatives then transform as
\begin{equation}
    f'(\xi) = \frac{1}{2} (z + 2 z') e^{-t/2}\,, \qquad f''(\xi) = -\frac{1}{4} (z - 4 z'')e^{-3t/2}\,.
\end{equation}
Under this substitution the ODE becomes autonomous
\begin{equation}
    4 z'' - z + c_{01} ^2 \left(z -2 z' \right)^3 = 0\,,
\end{equation}
and its order can be reduced via a substitution $p(z):=z'$
\begin{equation}
    4 p p' - z + c_{01}^2 (z -2 p)^3 = 0\,.
\end{equation}
The solutions to this ODE are 
\begin{equation}
    p = \frac{z}{2} + \frac{1}{2 c_{01}} \frac{-\hat{c}_{02} \, c_{01} z \pm \, \sqrt{\hat{c}_{02} (c_{01}^2 z^2 - 1) + 1}}{1 + \hat{c}_{02} \,c_{01}^2 z^2}
\end{equation}
for some constant $\hat{c}_{01}$. The special cases that correspond to $\hat{c}_{02} =1$ are
\begin{equation}\label{eq:app-irreg-p-sol}
    p = \frac{z}{2}\,, \qquad p = \frac{z}{2} \frac{c_{01}^2 z^2 - 1}{c_{01}^2 z^2 +1}\,.
\end{equation}
The function $z$ is found from solving the ODE $p(z):=z'$. Assuming $\hat{c}_{02} \neq 1$, one has two branches
\begin{gather}
    f = z e^{t/2} =  \pm \frac{\sqrt{\xi}}{c_{01}} \left[ 1 + \left(c_{03}^{-1}+\frac{1}{4} c_{02}^2 c_{03}\right) \xi + c_{02} \sqrt{\xi (\xi + c_{03})} \right]^{1/2}, \quad c_{03} \neq 0\\
    f = z e^{t/2} =  \pm \frac{\sqrt{\xi}}{c_{01}} \left[ 1 - \left(c_{03}^{-1}-\frac{1}{4} c_{02}^2 c_{03}\right) \xi + c_{02} \sqrt{\xi (c_{03}-\xi)} \right]^{1/2}, \quad c_{03} > 0
\end{gather}
where we introduce new parameters $c_{02}, c_{03}$ which satisfy
\begin{equation}
    c_{02} \,c_{03} = - 2 |\hat{c}_{02}|^{-1/2}\,,
\end{equation}
to match the notation of \cref{eq:twist-gen-irreg_c2-sol1,eq:twist-gen-irreg_c2-sol2}. For $\hat{c}_{02} =1$, the first solution gives $f$ a linear function of $\xi$, i.e. the case of \cref{sub:conf-orthotoric}. The second solution gives
\begin{equation}
    f = \pm \frac{1}{c_{01}} \sqrt{\xi (1 + c_{02} \xi)}
\end{equation}
where $c_{02}$ is a constant of integration.

\bibliographystyle{unsrt}
\bibliography{my_libv3}

\end{document}